\documentclass[aps,pra,twocolumn]{revtex4-2}%
\usepackage{graphicx}
\usepackage{amsthm}
\usepackage{physics}
\usepackage{tikz}
\usetikzlibrary{arrows.meta, positioning}
\usepackage{amsmath}
\usepackage{bm}
\usepackage{hyperref}
\usepackage{makecell}
\usepackage{mathrsfs}
\usepackage{ragged2e}
\usepackage{hyperref}

\usepackage{amsfonts}
\usepackage{caption}
\usepackage{subcaption}
\usepackage{amssymb}
\usepackage{cleveref}
\usepackage{booktabs}

\hypersetup{
    colorlinks=true,
    linkcolor=blue,
    citecolor=blue,
    urlcolor=blue
}

\def\W{\mathbb{W}}

\newcommand{\kb}[2]{\ketbra{#1}{#2}}

\newtheorem{thm}{Theorem}
\newtheorem{prop}{Proposition}

\newtheorem{eg}{Example}

\begin{document}
	
\preprint{APS/123-QED}
	
\title{Imaginarity witnessing enhancement via spectral norms of witnesses}
\author{Yuhang Xie}
\affiliation{School of Mathematics and Statistics, Henan University, Kaifeng, 475004, China}
\author{Yanjun Chu}
\email{chuyj@henu.edu.cn}
\affiliation{School of Mathematics and Statistics, Henan University, Kaifeng, 475004, China}
\author{Yushan Ding}
\affiliation{School of Mathematics and Statistics, Henan University, Kaifeng, 475004, China}
\author{Shao-Ming Fei}
\email{feishm@cnu.edu.cn}
\affiliation{School of Mathematics Sciences, Capital Normal University, Beijing, 100048, China}

\begin{abstract}
Quantum imaginarity is an essential physical resource that underpins key functionalities of modern quantum technologies. We improve imaginarity witnessing via the prior knowledge of imaginarity‑witness operators. To this end, we derive a rigorous upper bound on the maximum expectation value of an imaginarity witness operator over the set of all free (real) quantum states, which is given by the spectral norm of the real component of the corresponding witness operator. We demonstrate via detailed examples that this bound substantially improves imaginarity detection. We further classify all imaginarity‑witness operators into four distinct families based on this bound. For these four witness classes, we perform a comprehensive analysis of their completeness and finite completeness, the joint detection of shared imaginary quantum states by different witnesses, and the conditions for distinct witnesses to identify identical imaginary states. Our results advance the fundamental understanding of imaginarity detection and offer useful insights for both theoretical studies and experimental implementations of quantum imaginarity.
\end{abstract}

\maketitle

\section{Introduction}

Imaginary numbers lie at the heart of quantum theory, providing the indispensable mathematical framework for describing quantum states and their dynamical evolution, and constitute a ubiquitous tool throughout both classical and quantum physics. A series of recent theoretical proposals and experimental verifications have firmly established that complex Hilbert spaces are unavoidable for a consistent formulation of quantum mechanics \cite{renou2021quantum,li2022testing,chen2022ruling,wu2022experimental}. The imaginary component inherent to complex quantum descriptions further gives rise to quantum imaginarity, a distinct quantum resource whose practical and foundational implications have been explored across a wide range of quantum information tasks. Prior work has uncovered key roles of imaginarity in quantum hiding and masking \cite{zhu2021,zhang2021}, multiparameter quantum metrology \cite{miyazaki2022}, quantum machine learning \cite{sajjan2023}, quantum pseudorandomness \cite{haug2025}, output statistics of linear optical setups \cite{jones2023}, Kirkwood–Dirac quasiprobability distributions \cite{budiyono2023108,budiyono2023107,budiyono2023jpa,wagner2024}, weak-value formalism \cite{wagner2023}, nonlocal quantum advantages rooted in imaginarity \cite{weiPRA052202}, and imaginarity-based quantum speed limits \cite{xuanPRA0522002}.

Hickey and Gour pioneered the resource theory of quantum imaginarity in 2018 to systematically characterize the fundamental role of complex numbers in quantum mechanics \cite{hickey2018quantifying}. This theoretical framework offers a rigorous paradigm for the quantitative evaluation of state imaginarity, laying a solid foundation for subsequent investigations of imaginarity as a bona fide quantum resource \cite{chitambar2019}.
The quantitative characterization of imaginarity constitutes a core fundamental problem within resource theory. To date, a variety of imaginarity quantifiers have been established, many of which possess clear operational interpretations. Representative quantifiers include the robustness of imaginarity \cite{hickey2018quantifying,wu2021pra}, imaginarity fidelity \cite{wu2021prl}, and geometric imaginarity \cite{wu2021pra}, together with various other imaginarity quantifiers \cite{XueQIP383,chen220311,XuPLA130024,ZhangQIP361,chen280312,
LiuJPA035302,WuCTP095101,TianPLA130479,ZhangEPJP669,DuPRA022405}. Such interpretations not only facilitate a physically intuitive understanding of imaginarity as a quantum resource, but also clarify the inherent advantages of imaginary quantum states over real ones in diverse quantum protocols. Furthermore, by exploiting the state-channel duality in quantum mechanics, the concept of state imaginarity has been further extended to quantum channels, establishing a generalized framework for channel imaginarity characterization \cite{ChenPLA130129,Wu230316,WangEPJP216,FuQIP140}.
Besides quantitatively characterizing the amount of resource contained in a quantum state of a physical system, revealing the properties of an unknown quantum state in a physical system is also an important problem in quantum information theory. For instance, one may wish to determine whether a quantum state is entangled or separable, coherent or incoherent, as well as imaginary or real. Recently, Ref.~\cite{Chakrabarty00164} developed a framework for detecting imaginarity using moments of the extended Kirkwood-Dirac quasiprobability distribution.
By quantum state tomography, the full information of a given physical system can be determined through a series of quantum measurements. Nevertheless, this is generally an expensive and cost-inefficient procedure that consumes substantial experimental resources. Encouragingly, witness operators offer feasible solutions to these problems using currently accessible techniques, e.g. entanglement witnesses \cite{LewensteinPRA052310,HouPRA052301,HorodeckiPRA865,WangPRA050302R} and coherence witnesses \cite{MaPRA012429,WangQIP181,WangEntropy1136}. Witness operators typically require less experimental setups to unambiguously verify the existence of resource. This concept builds upon the convexity and closedness of the set of free states, by the Hahn-Banach theorem and Riesz representation theorem \cite{functionalanalysis},  there exists a Hermitian operator $W$ that can be used to detect some resource states $\rho$, namely, $\Tr[W\rho]<0$. Geometrically, witness operators correspond to hyperplanes separating such resource states from the set of free states. It is worth noting that the value $\Tr[W\rho]$ obtained from a witness operator can also be employed to define resource measures \cite{EisertNJP46,Ren281}. This unveils that resource witness constitutes a highly effective and physically implementable tool for characterizing quantum states.

In analogue to the entanglement witness and coherence witness, the imaginarity witness was proposed for experimental detection of quantum imaginarity \cite{zhangPLA130135}. Specifically, a Hermitian operator $W$ serves as a valid imaginarity witness if it satisfies two conditions: (i) $\Tr[W\sigma] \geqslant 0$ for all real quantum states $\sigma$; (ii) there exists at least one imaginary state $\rho$ such that $\Tr[W\rho] < 0$. Zhang \textit{et al} also employed imaginarity witnesses to derive explicit lower bounds for the robustness and $\ell_1$-norm of imaginarity, revealing an intrinsic correspondence between imaginarity witnesses and imaginarity quantifiers, highlighting the pivotal role of witness operators in probing imaginarity properties. Furthermore, Ref.~\cite{FernandesPRL190201} utilized a kind of unitary invariants, termed Bargmann invariants, as witnesses for quantum imaginarity. 

Inspired by Ref.~\cite{LiPRA032422}, we aim to obtain a suitable upper bound $S$ for the expectation value of $W$ over all real states. In this way, whenever the measurement outcome violates the inequality $ 0\leqslant \Tr[W\rho] \leqslant S$, we conclude that $\rho$ is a imaginary state. We systematically investigate imaginarity witnesses within the framework of prior knowledge about the spectral norm of the real parts of the witness observables. Through explicit examples, we demonstrate that the bound $S$ indeed enhances the ability of witnesses in detecting imaginary states. We then classify all witness operators with respect to their bounds. This paper is organized as follows. In Sec.~\ref{sect2}, we derive a rigorous upper bound on the expectation value of any Hermitian operator in all real quantum states, and show that this bound substantially enhances the detection capability of witness operators. In Sec.~\ref{sect3}, we categorize all imaginarity witnesses into four families according to the magnitude of the spectral norm. We then systematically investigate the intrinsic properties of these four witness classes, focusing on their completeness and finite completeness, as well as on the conditions under which different witness operators identify the same imaginary states. In Sec.~\ref{sect4}, we briefly compare our results with previous studies on coherence witnessing.  Concluding remarks are presented in Sec.~\ref{sect5}.

\section{Witnessing Quantum Imaginarity}\label{sect2}

Let $\mathscr{H}$ denote a $d$-dimensional Hilbert space equipped with the computational basis $\mathcal{B}=\{\ket{k}\mid k=1,2,\dots,d\}$. Denote $\mathbb{H}$ the set of $d\times d$ Hermitian operators and $\mathcal{D}$ the set of density operators in $\mathscr{H}$. Let $\mathcal{R}$ stand for the set of real quantum states defined with respect to the basis $\mathcal{B}$, namely, the free states in the resource theory of quantum imaginarity. Then the states in the complement set $\mathcal{D}\setminus\mathcal{R}$ are said to be imaginary (resource) states. The real and imaginary parts of a quantum state $\rho$ with respect to the reference basis $\mathcal{B}$ are respectively given by
\begin{equation}\label{Re and Im}
\Re(\rho):=\frac{\rho+\rho^\top}{2},~~ \quad \Im(\rho):=\frac{\rho-\rho^\top}{2\mathbf{i}},
\end{equation}
where $\mathbf{i}=\sqrt{-1}$ is the imaginary unit and $\rho^\top$ denotes the transpose of $\rho$.
We herein define $\mathbb{H}^{\text{Re}}_\geqslant$ as the set of Hermitian operators with positive semi-definite real parts and $\Delta_-$ as the set of Hermitian operators admitting negative eigenvalues. The full set of general imaginarity witnesses is thus given by the intersection $\mathbb{H}^{\text{Re}}_\geqslant\bigcap\Delta_-$. In this manuscript, we adopt $\|X\|_{\infty}=\lim_{p\rightarrow\infty}\|X\|_p$ to denote the spectral norm of an operator $X$, where
\[
\|X\|_p=\left\{\Tr\Bigl[\bigl(\sqrt{X^\dagger X}\bigr)^p\Bigr]\right\}^{1/p}.
\]
The spectral norm of $X$ equals its largest singular value numerically \cite{WildeBook}.

\begin{thm}\label{thmmain}
For any real state $\sigma$, we obtain the upper bound
\[
\mathrm{Tr}[W\sigma]\leqslant \norm{\operatorname{Re}(W)}_\infty,
\]
where $\norm{\operatorname{Re}(W)}_\infty$ denotes the spectral norm of $\operatorname{Re}(W)$.
\end{thm}

The proof of Theorem 1 is straightforward. By the identity \(\mathrm{Tr}[W\sigma]=\mathrm{Tr}[W^\top\sigma]\), for any real state \(\sigma\), we have 
\begin{align*}
\mathrm{Tr}[W\sigma]&=\frac{\mathrm{Tr}[W\sigma]+\mathrm{Tr}[W^\top\sigma]}{2}= \mathrm{Tr}[\Re(W)\sigma] \\
&\leqslant\norm{\Re(W)}_\infty\norm{\sigma}_1 = \norm{\Re(W)}_\infty,
\end{align*}
where we have used the H\"older inequality \(|\mathrm{Tr}[A^\dagger B]| \leqslant \|A\|_p\|B\|_q\) for \(p,q>0\) satisfying \(1/p + 1/q = 1\). 

Theorem 1 says that any state \(\rho\) obeying \(\mathrm{Tr}[W\rho]>\norm{\Re(W)}_\infty\) must be an imaginary state. This implies that the knowledge of the spectral norm of the real part of the witness operator broadens the set of detectable imaginary states, as visualized in FIG.~\ref{fig1}. Specifically, imaginary states situated above the solid black curve are distinguishable using the traditional criterion \(\mathrm{Tr}[W\rho]<0\), while certain imaginary states lying below the black solid line can be detected through our new criterion \(\mathrm{Tr}[W\rho]>\norm{\Re(W)}_\infty\).

\begin{figure}[htbp]
    \centering
    \includegraphics[width=0.5\textwidth,trim=25 60 25 25, clip]{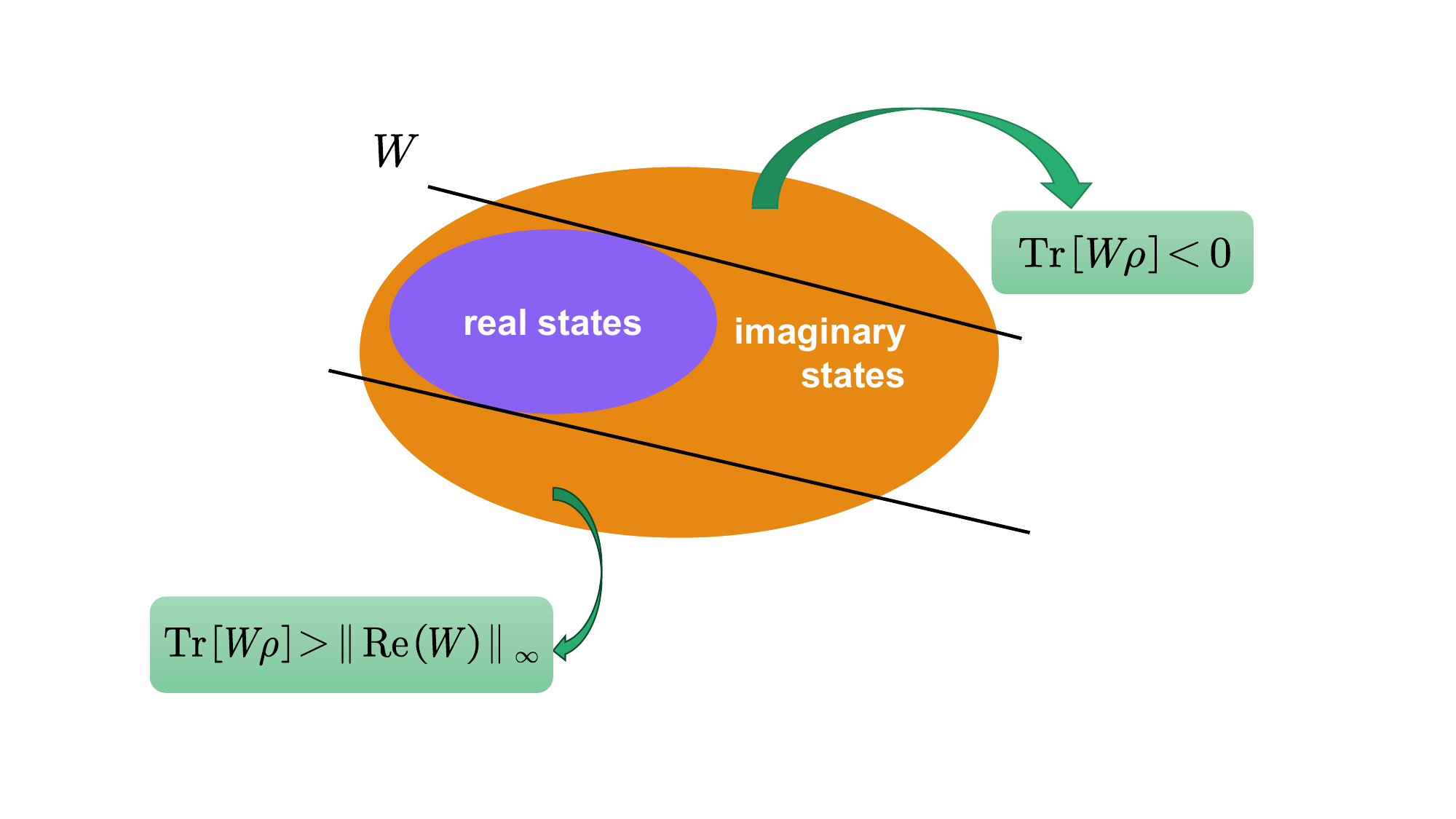}
    \caption{Classification of detectable imaginary quantum states via the imaginarity witness \(W\): those satisfying $\Tr[W\rho]<0$ and those satisfying $\Tr[W\rho]>\norm{\Re(W)}_{\infty}$.}
    \label{fig1}
\end{figure}

To characterize the collection of imaginarity witnesses with an identical spectral norm, we define
\begin{equation}\label{prior knowledge}
\mathbb{W}_S:=\left\{W \in \mathbb{H}^{\text{Re}}_\geqslant \, \bigg| \, \norm{\Re(W)}_\infty \leqslant S\right\}
\end{equation}
for a non-negative real number \(S\).
Considering the convex combination $\sum_k p_k W_k$ of finitely many Hermitian operators $\{W_k\}\subseteq\mathbb{W}_S$, the following inequality holds \cite{MatrixAnalysis},
\begin{equation*}
\norm{\sum_k p_k W_k}_\infty\leqslant\sum_k p_k \norm{W_k}_\infty\leqslant\sum_k p_k S =S,
\end{equation*}
which implies that $\sum_k p_k W_k\in\mathbb{W}_S$, so  $\mathbb{W}_S$ forms a convex set.

In \cite{zhangPLA130135} the authors proposed an alternative class of witnesses termed stringent imaginarity witnesses, satisfying \(\Tr[W\sigma] = 0\) for every \(\sigma\in\mathcal{R}\) and \(\Tr[W\rho] \neq 0\) for at least one imaginary state \(\rho\in\mathcal{D}\setminus\mathcal{R}\). For a stringent imaginarity witness \(W\), \(\Tr[W\sigma] = 0\) for all real states \(\sigma\in\mathcal{R}\) is equivalent to \(\text{Re}(W) = 0\).  It  follows that the stringent imaginarity witnesses must be chosen from the set $\mathbb{W}_0=\left\{W \in \mathbb{H}^{\text{Re}}_\geqslant \, \mid \, \|\Re(W)\|_\infty=0\right\}$.

From Theorem 1, a witness operator \(W\) equipped with prior knowledge satisfying \(\|\Re(W)\|_\infty=S>0\) can detect the imaginarity of a state \(\rho\) via either \(\Tr[W\rho]<0\) or \(\Tr[W\rho]>S\). This significantly enhances the conventional imaginarity‑detection capability of imaginarity witnesses; see FIG.~\ref{fig2}.

\begin{figure}[htbp]
    \centering
    \includegraphics[width=0.5\textwidth]{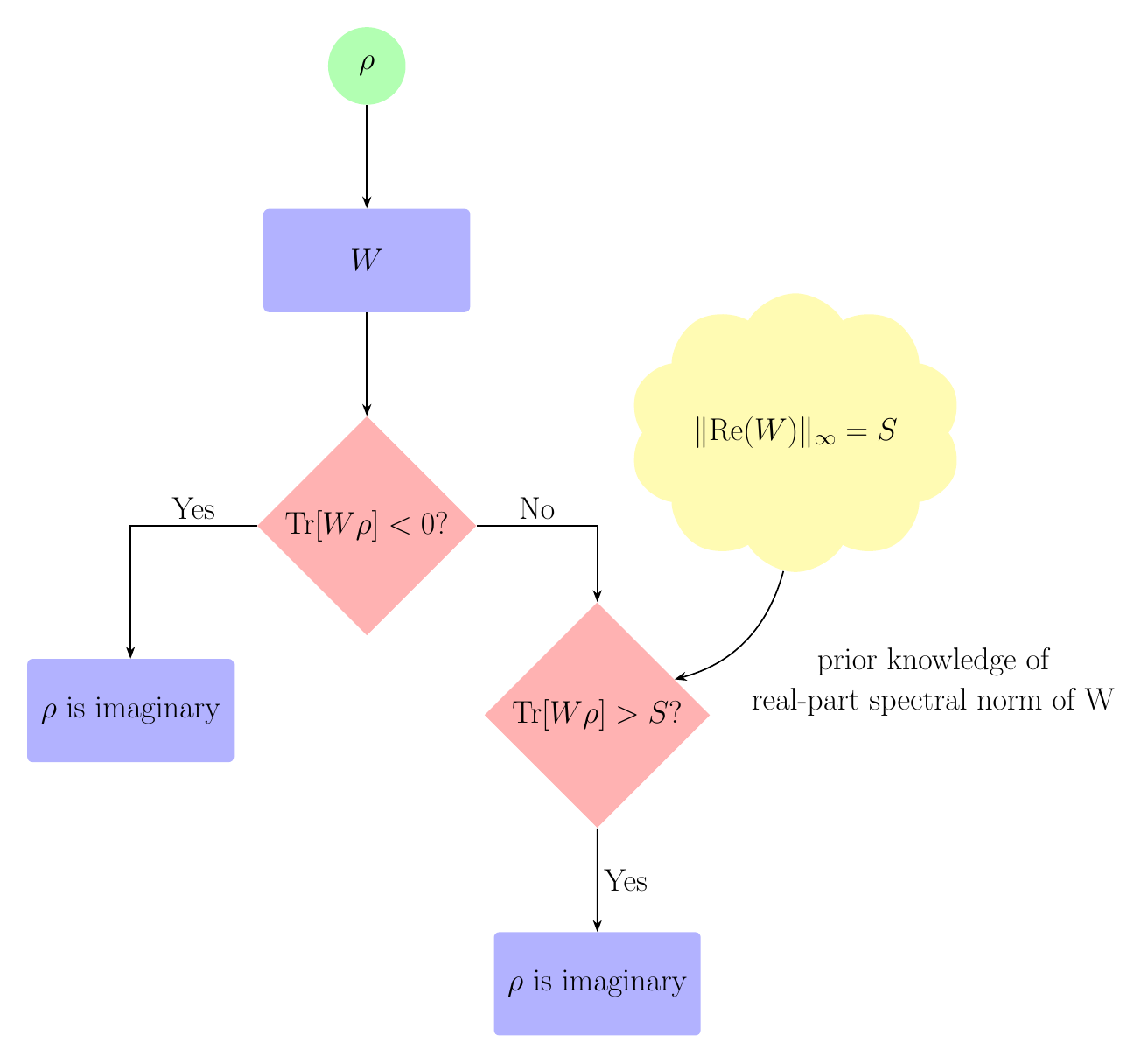}
    \caption{Witnessing imaginarity with prior knowledge of the real-part spectral norm of observable $W$.}
    \label{fig2}
\end{figure}

The existence of $W\in \mathbb{W}_S$ and $\rho$ such that $\Tr[W\rho] >S$ can be illustrated in the following way.
 Define
\begin{equation*}
W = S\ketbra{j}{j} - \frac{1}{2}\sqrt{s}\mathbf{i}\ketbra{j}{k} + \frac{1}{2}\sqrt{s}\mathbf{i}\ketbra{k}{j}
\end{equation*}
for any \(s>0\) and \(1\leqslant j\neq k\leqslant d\). \(W\) has distinct nonzero simple eigenvalues $\mu_1=\frac{S+\sqrt{S^2+s}}{2}>0$ and $\mu_2=\frac{S-\sqrt{S^2+s}}{2}<0$. Let \(|\phi_1\rangle\) and \(|\phi_2\rangle\) be the corresponding eigenvectors associated with $\mu_1$ and  $\mu_1$, respectively. It follows that \(W\) admits a spectral decomposition \(W = \mu_1\ketbra{\phi_1}{\phi_1} + \mu_2\ketbra{\phi_2}{\phi_2}\). For state $\rho=\ketbra{\phi_1}{\phi_1}$, we have \(\Tr[W\rho]=\mu_1=\frac{S+\sqrt{S^2+s}}{2}>S.\) Therefore, $W$ detects the imaginarity of the quantum state $\rho$.

Furthermore, note that an imaginarity witness $W$ is constructed from $\mathbb{H}^{\text{Re}}_{\geqslant}\bigcap\Delta_{-}$ \cite{zhangPLA130135}, where $W$ must have at least one negative eigenvalue to ensure the existence of some imaginary state $\rho$ satisfying $\Tr[W\rho]<0$. Namely, within the conventional imaginarity witness framework, no positive semi-definite Hermitian operator $W$ can detect the imaginarity of any quantum states. In what follows, we present an example to show that a positive semi-definite observable $W$ is still capable of detecting the imaginarity of certain class of quantum states, provided that the spectral norm of the real part \(\norm{\Re(W)}_{\infty}\) of $W$ is known in advance. This reveals that several Hermitian operators that are invalid as imaginarity witnesses in conventional settings can still act as valid witnesses in this framework.

\begin{eg}
Consider the following states,
\begin{equation*}
W_p=\begin{pmatrix}
2 & -\mathbf{i} & 0 \\
\mathbf{i} &\ \  2 & 0 \\
0 &\ \  0 & p
\end{pmatrix}, \quad
\rho=\begin{pmatrix}
\frac{1}{3} & -\frac{\mathbf{i}}{3} & 0 \\[3pt]
\frac{\mathbf{i}}{3} &\ \  \frac{1}{3} & 0 \\[3pt]
0 & \ \ 0 & \frac{1}{3}
\end{pmatrix},
\end{equation*} where $p\in [0,+\infty).$
It is verified that both \(\Re(W_p)\) and $W_p$ are positive semi-definite.
The spectral norm of $\Re(W_p)$ is given by
\begin{equation*}
\|\Re(W_p)\|_\infty =
\begin{cases}
2, & 0 \leqslant p \leqslant 2,\\
p, & 2< p <+\infty .
\end{cases}
\end{equation*}
Since $\Tr[W_p\rho]=\dfrac{p+6}{3}>0$, $W_p$ detects no imaginarity of $\rho$ in the conventional way. However,
since $\Tr[W_p\rho]=\dfrac{p+6}{3}>\|\Re(W_p)\|_\infty$ for $0< p<3$, $W_p$ detects the imaginarity of $\rho$ whenever $0< p<3$, 
see FIG.~\ref{figeg1}.

\begin{figure}[htbp]
    \centering
    \includegraphics[width=0.4\textwidth]{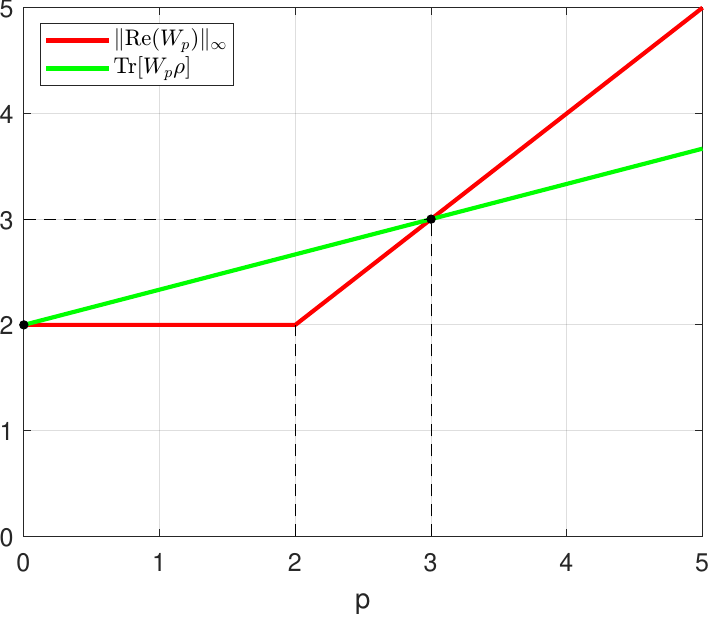}
    \caption{$\|\Re(W_p)\|_{\infty}$ and  $\Tr[W_p\rho]$ versus $p$. $\Tr[W_p\rho]>\|\Re(W_p)\|_{\infty}$ for $0< p<3$.}
    \label{figeg1}
\end{figure}
\end{eg}

\begin{eg}
Consider a family of quantum states $\ket{\psi(\theta,t)}$ and Hermitian operators $W(\alpha)$ given by:
\begin{equation*}
\ket{\psi(\theta,t)}=\frac{1}{\sqrt{1+t^2}}
\begin{pmatrix}
\cos\theta \\
\mathbf{i}\sin\theta \\
\mathbf{i}t
\end{pmatrix}
\end{equation*}
\begin{equation*}
W=
\begin{pmatrix}
1 & \mathbf{i}\alpha & \mathbf{i}\alpha \\
-\mathbf{i}\alpha & 1 & \mathbf{i}\alpha \\
-\mathbf{i}\alpha & -\mathbf{i}\alpha & 1
\end{pmatrix}\\
:=\mathbb{I}_3+\mathbf{i}\alpha B.
\end{equation*}
Here, $\Re(W(\alpha))$ is positive semi-definite with $\norm{\Re(W(\alpha))}_{\infty}=1$.
The eigenvalues of $W(\alpha)$ are $1,1-\sqrt{3}\alpha,1+\sqrt{3}\alpha$. 
$W(\alpha)$ has a negative eigenvalue if and only if either $\alpha>\frac{1}{\sqrt{3}}$ or $\alpha<-\frac{1}{\sqrt{3}}$.
Direct computation yields
\begin{equation*}
\begin{split}
\Tr[W(\alpha)\kb{\psi(\theta,t)}{\psi(\theta,t)}]&=\Tr[(\mathbb{I}_3+\mathbf{i}\alpha B)\kb{\psi(\theta,t)}{\psi(\theta,t)}] \\
&=1-\frac{2\alpha}{1+t^2}(\cos\theta\sin\theta+t\cos\theta).
\end{split}
\end{equation*}
When $t=1$, we have $$
\begin{array}{llll}
f(\alpha,\theta)&\equiv \Tr[W(\alpha)\kb{\psi(\theta,1)}{\psi(\theta,1)}]\\
&=1-\alpha(\cos\theta\sin\theta+\cos\theta)\\
&=1-\alpha g(\theta),\end{array}$$ where
$g(\theta)=\cos\theta\sin\theta+\cos\theta$. It is readily shown that $g(\theta)$ is monotonically increasing on $\left[-\frac{7\pi}{6}+2k\pi,\frac{\pi}{6}+2k\pi\right]$ and monotonically decreasing on $\left[\frac{\pi}{6}+2k\pi,\frac{5\pi}{6}+2k\pi\right]$,  with range $g(\theta)\in\left[-\frac{3\sqrt{3}}{4},\frac{3\sqrt{3}}{4}\right]$.

(1) When $\alpha>\frac{1}{\sqrt{3}}$, $f(\alpha,\theta)<0$ implies that $g(\theta)>\frac{1}{\alpha}$. In this case, $\Tr[W(\alpha)\kb{\psi(\theta,t)}{\psi(\theta,t)}]<0$ necessary requires that $\frac{1}{\alpha}<\frac{3\sqrt{3}}{4}$, i.e., $\alpha>\frac{4}{3\sqrt{3}}$, namely, for $\frac{1}{\sqrt{3}}<\alpha\leqslant\frac{4}{3\sqrt{3}}$, 
the conventional approach cannot detect the imaginarity of the state $\ket{\psi(\theta,1)}$. 
Moreover, from the properties of $g(\theta)$, there exist $\theta_1\in\left(-\frac{\pi}{2},\frac{\pi}{6}\right)$ and $\theta_2\in\left(\frac{\pi}{6},\frac{\pi}{2}\right)$ such that $g(\theta_1)=g(\theta_2)=\frac{1}{\alpha}$. Therefore, the parameter $\theta$ corresponding to detectable imaginary states lies in $\left(\theta_1+2k\pi,\theta_2+2k\pi\right)$. When $f(\alpha,\theta)>\norm{\Re(W(\alpha))}_{\infty}=1$, the condition $g(\theta)<0$ suffices. In this scenario, the parameter range of detectable imaginary states is $\theta\in\left(\frac{\pi}{2}+2k\pi,\frac{3\pi}{2}+2k\pi\right)$.

(2)When $\alpha<-\frac{1}{\sqrt{3}}$ and $f(\alpha,\theta)<0$, we obtain $g(\theta)<\frac{1}{\alpha}$. This immediately imposes the necessary condition $\frac{1}{\alpha}>-\frac{3\sqrt{3}}{4}$, or equivalently $\alpha<-\frac{4}{3\sqrt{3}}$ for $\Tr[W(\alpha)\kb{\psi(\theta,t)}{\psi(\theta,t)}]<0$. As there exist angles $\theta_1\in\left(\frac{\pi}{2},\frac{5\pi}{6}\right)$ and $\theta_2\in\left(\frac{5\pi}{6},\frac{3\pi}{2}\right)$ satisfying $g(\theta_1)=g(\theta_2)=\frac{1}{\alpha}$, the values of $\theta$ corresponding to detectable imaginary states belong to the intervals $\left(\theta_1+2k\pi,\theta_2+2k\pi\right)$. In the regime $f(\alpha,\theta)>\norm{\Re(W(\alpha))}_{\infty}=1$, we have $g(\theta)>0$. Hence, the parameter range yielding detectable imaginary states is $\theta\in\left(-\frac{\pi}{2}+2k\pi,\frac{\pi}{2}+2k\pi\right)$.

We visualize these results in FIG.~\ref{figeg2}. If the parameter $\theta$ lies in the green sector, the imaginary states are detected by $\Tr[W(\alpha)\kb{\psi(\theta,1)}{\psi(\theta,1)}]<0$ . If $\theta$ falls within the blue sector, the imaginary state is detectable by $\Tr[W(\alpha)\kb{\psi(\theta,1)}{\psi(\theta,1)}]>\norm{\Re(W)}_{\infty}$.

\begin{figure}[htbp]
    \centering
    \includegraphics[width=0.5\textwidth]{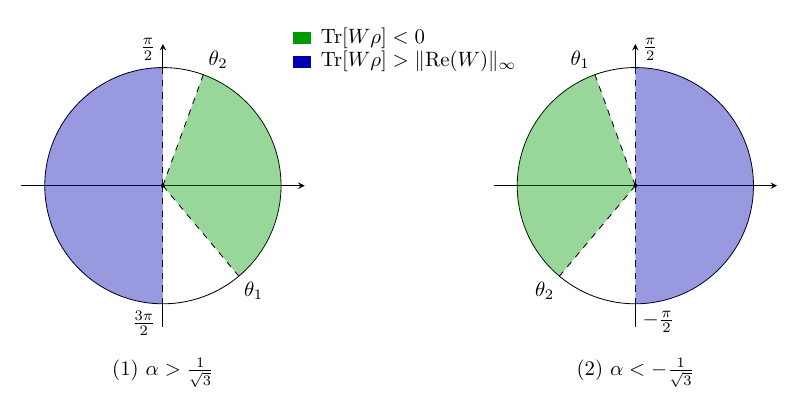}
    \caption{Imaginary states detected by $W(\alpha)$ for $\alpha>\frac{1}{\sqrt{3}}$ and $\alpha<-\frac{1}{\sqrt{3}}$.}
    \label{figeg2}
\end{figure}
\end{eg}


We now present a proposition characterizing when there exists an imaginary state $\rho$ satisfying $\Tr[W\rho]>\norm{\Re(W)}_{\infty}$ for a Hermitian operator $W$.

\begin{prop}
For a given $W\in\mathbb{W}_S$, there exists an imaginary state $\rho$ such that $\Tr[W\rho]>S$ if and only if $\|W\|_{\infty}>\|\Re(W)\|_{\infty}$.
\end{prop}

\begin{proof}
We first establish the inequality $\|W\|_{\infty}\geqslant\|\Re(W)\|_{\infty}$. Recall that the spectral norm satisfies $\|W\|_{\infty}=\lambda_{max}(W)$ and $\|\Re(W)\|_{\infty}=\lambda_{max}(\Re(W))$.
By definition,
\begin{equation*}
\begin{aligned}
\lambda_{max}(W)
&=\max_{\substack{x\in\mathbb{C}^n\\\|x\|=1}}x^\dagger W x
\geqslant \max_{\substack{x\in\mathbb{R}^n\\\|x\|=1}}x^\dagger W x \\
&=\max_{\substack{x\in\mathbb{R}^n\\\|x\|=1}}x^{\top} W x
=\max_{\substack{x\in\mathbb{R}^n\\\|x\|=1}}x^{\top}\Re(W)x \\
&=\lambda_{max}(\Re(W)),
\end{aligned}
\end{equation*}
where the equality holds if and only if there exists a real normalized vector $x\in\mathbb{R}^n$ such that $Wx=\lambda_{max}(W)x$.

We now prove the claimed equivalence. Suppose $\|W\|_{\infty}\leqslant\|\Re(W)\|_{\infty}$. From the above arguments, we obtain $\|W\|_{\infty}=\|\Re(W)\|_{\infty}$.
For any imaginary state $\rho$, one has $\Tr[W\rho]\leqslant\|W\|_{\infty}=\|\Re(W)\|_{\infty}$. Consequently, there exist no imaginary states for which \(\Tr[W\rho]>S\).

Assume that  $\|W\|_{\infty}>\|\Re(W)\|_{\infty}$.
Let $\ket{\psi}$ be a normalized eigenvector of $W$ corresponding to its maximum eigenvalue $\lambda_{max}(W)$.
By the condition $\|W\|_{\infty}>\|\Re(W)\|_{\infty}$, $\ket{\psi}$ cannot be a real vector.
Define
$\rho_\epsilon=(1-\epsilon)\ketbra{\psi}{\psi}+\epsilon\sigma,\quad \epsilon\in[0,1]$, where $\sigma$ is an arbitrary real state.
For sufficiently small $\epsilon>0$, $\rho_\epsilon$ is an imaginary state. Direct computation yields
\begin{equation*}
\begin{split}
\Tr[W\rho_{\epsilon}]
&=(1-\epsilon)\Tr\big[W\ketbra{\psi}{\psi}\big]+\epsilon\Tr[W\sigma] \\
&=(1-\epsilon)\lambda_{max}(W)+\epsilon\Tr[W\sigma].
\end{split}
\end{equation*}
Since $\lambda_{max}(W)>\|\Re(W)\|_{\infty}= S\geqslant \Tr[W\sigma]\geqslant 0$, we have
$\Tr[W\rho_{\epsilon}]>(1-\epsilon)\lambda_{max}(W)$.
Selecting a sufficiently small $\epsilon=1-\frac{S}{\lambda_{max}(W)}>0$ such that $\Tr[W\rho_\epsilon]>S$, we prove that such an imaginary state \(\rho=\rho_\epsilon\) exists.
\end{proof}

However, in contrast to the preceding scenario, where the spectral norm of the real part of the observable is bounded above by a positive constant, we may encounter cases in which the spectral norm of the real part of these Hermitian operators is merely positive, or equivalently, the operators are selected from the set \(\mathbb{W}_>:=\{W\in\mathbb{H}^{\text{Re}}_\geqslant\mid\|\operatorname{Re}(W)\|_{\infty}>0\}.\) With only this limited prior knowledge about the observable, no enhancement in imaginarity detection can be attained. To verify this, we only need to show that
\begin{equation*}
\left\{\Tr[W\sigma]\mid W\in\mathbb{W}_>,\sigma\in\mathcal{R}\right\}=\left\{s\in\mathbb{R}\mid s\geqslant 0\right\}.
\end{equation*}
That is, for any real number \(s \geqslant 0\), there exist \(W\in\mathbb{W}_>\) and \(\sigma\in\mathcal{R}\) such that \(\Tr[W\sigma]=s\). If $s>0$, choosing $W=s\mathbb{I}+\sqrt{d}s\mathbf{i}(-\ketbra{p}{q}+\ketbra{q}{p})$ with $1 \leqslant p < q \leqslant d$ and $\sigma=\frac{\mathbb{I}}{d}+b_{jk}(\ketbra{j}{k}+\ketbra{k}{j})$ with $1 \leqslant j < k \leqslant d$, $\abs{b_{jk}}\leqslant\frac{1}{d}$, so we have $\Tr[W\sigma]=s+\sqrt{d}rb_{jk}\mathbf{i}(-\delta_{qj}\delta_{pk}-\delta_{qk}\delta_{pj}+\delta_{pj}\delta_{qk}+\delta_{pk}\delta_{qj})=s$. If $s=0$, setting $W=\ketbra{j}{j}-\sqrt{d}\mathbf{i}\ketbra{j}{k}+\sqrt{d}\mathbf{i}\ketbra{k}{j}$ and $\sigma=\ketbra{k}{k}$, we have $\Tr[W\sigma]=0=s$. Consequently, to ensure that the witness operator $W$ detects the imaginarity of $\rho$, we still need to observe the violation of the inequality $\Tr[W\rho]<0$ as usual.

Similarly, we collect all Hermitian operators in $\mathbb{H}^{\text{Re}}_\geqslant$ such that that $\norm{\Re(W)}_\infty \geqslant 0$, i.e., $\mathbb{W}_\geqslant:=\{W\in\mathbb{H}^{\text{Re}}_\geqslant\mid\norm{\Re(W)}_\infty\geqslant0\}$. From the above discussion, it follows that the relation $\{\Tr[W\sigma]\mid W\in\mathbb{W}_\geqslant,\sigma\in\mathcal{R}\}=\{s\in\mathbb{R}\mid s\geqslant 0\}$ also holds.


Let $S$ be a positive real number, we introduce the notations $\mathbb{E}_S^\ell[W]:=\{\rho \in \mathcal{D} \mid \Tr[W\rho] <0\}$ and $\mathbb{E}_S^r[W]:=\{\rho \in \mathcal{D} \mid \Tr[W\rho] >S\}$. Accordingly, we have $\mathbb{E}_S[W]=\mathbb{E}_S^\ell[W]\bigcup\mathbb{E}_S^r[W]$. With these notations, we present the relationship between the traditional witness method and our proposed method.

\begin{prop}\label{dual witnesses}
Let $W\in\mathbb{W}_S$ with $S$ being a positive real number, we have the following conclusion.
\begin{enumerate}
\item[1)] If $\rho$ is an imaginary state satisfying $\Tr[W\rho]>S\geqslant\norm{\Re(W)}_{\infty}$, i.e. $\rho\in\mathbb{E}_S^r[W]$, then there exists a witness operator $\widetilde{W}$ such that $\Tr[\widetilde{W}\rho]<0$, namely, $\rho\in\mathbb{E}_\geqslant\left[\widetilde{W}\right]$.
\item[2)] If $\rho$ is an imaginary state satisfying $\Tr[W\rho]<0$, i.e. $\rho\in\mathbb{E}_{\geqslant}[W]$, then there exists a positive number $S'$ and a corresponding witness operator $\widehat{W}\in\mathbb{W}_{S'}$ such that $\Tr[\widehat{W}\rho]>S'$, namely, $\rho\in\mathbb{E}_{S'}^r\left[\widehat{W}\right]$.
\end{enumerate}
\end{prop}

\begin{proof}
First, suppose $W\in\mathbb{W}_S$ and $\Tr[W\rho]>S$. Define $\widetilde{W}:=S\mathbb{I}_d-W$. It is verifies that $\Re\left(\widetilde{W}\right)\geqslant 0$ and $\Tr[\widetilde{W}\rho]=S-\Tr[W\rho]<0$, which implies $\rho\in\mathbb{E}_\geqslant\left[\widetilde{W}\right]$.

Second, suppose $W\in\mathbb{W}_\geqslant$ and $\Tr[W\rho]<0$. Define $\widehat{W}=S'\mathbb{I}_d-W$, where $S'$ is any real number no less than $\norm{\Re(W)}_{\infty}$. Then $\Re\left(\widehat{W}\right)$ is positive semi-definite and $\norm{\Re\left(\widehat{W}\right)}_{\infty}\leqslant S'$, which yields $\widehat{W}\in\mathbb{W}_{S'}$. Since $\Tr[\widehat{W}\rho]=S'-\Tr[W\rho]>S'$, we have $\rho\in\mathbb{E}_{S'}^r\left[\widehat{W}\right]$.
\end{proof}

Proposition~\ref{dual witnesses} indicates that, given prior knowledge of the real‑part spectral norm $\norm{\Re(W)}_{\infty}$ of an observable $W$, all additional imaginary states detected by $W$ within our framework can be certified via a corresponding witness $\widetilde{W}$ in the conventional imaginarity witness paradigm. Conversely, imaginary states detected by the conventional method via negative expectation values can also be detected within our framework by means of a witness operator $\widehat{W}$ when its expectation values exceeding $\|\widehat{W}\|_{\infty}$. This is visualized in FIG.~\ref{figmirror}.

\begin{figure}[htbp]
    \centering
    \includegraphics[width=0.4\textwidth,trim=25 40 25 25, clip]{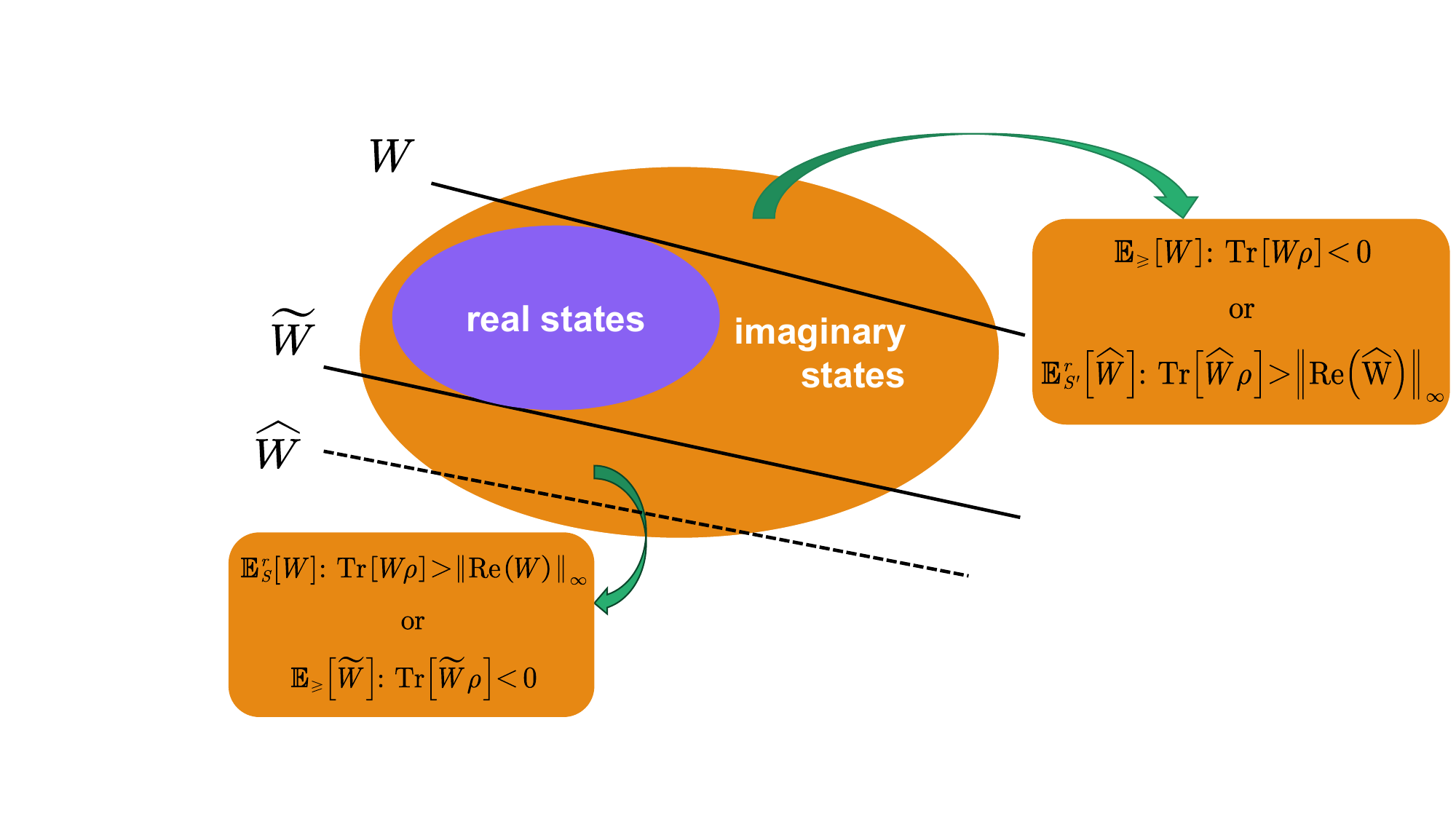}
    \caption{Relationship between the two criteria for detecting imaginary states: negative expectation values and expectation values exceeding the real‑part spectral norm of the observable.
}
    \label{figmirror}
\end{figure}

It follows from Proposition~\ref{dual witnesses} that for any $W \in \mathbb{H}^{\Re}_\geqslant$, if we set $\norm{\Re(W)}_\infty \leqslant S$ with $S$ being a nonnegative number, we obtain
\begin{equation}\label{leftright}
\mathbb{E}_S[W] = \mathbb{E}_\geqslant[W] \bigcup \mathbb{E}_\geqslant[S\mathbb{I}_d - W].
\end{equation}
It is easy to check that $\mathbb{E}_{\geqslant}[W]$ is a  convex set, i.e., if $\Tr[W\rho_1]<0$ and $\Tr[W\rho_2]<0$, then $\Tr[W((1-t)\rho_1+t\rho_2)]<0$ for all $t\in[0,1]$. However, if $\mathbb{E}_{\geqslant}[W]\neq\emptyset$ and $\mathbb{E}_\geqslant[S\mathbb{I}_d-W]\neq\emptyset$, the set $\mathbb{E}_S[W]$ is not convex. In this case, $\mathbb{E}_S[W]$ is a disjoint union of two convex sets. In fact, for any $\rho_1\in\mathbb{E}_\geqslant[W]$ and $\rho_2\in\mathbb{E}_\geqslant[S\mathbb{I}_d-W]$, we have both $\Tr[W\rho_1]<0$ and $\Tr[W\rho_2]>S$. Therefore, there are some $t^*\in(0, 1)$ such that $0<\Tr[W((1-t^*)\rho_1+t^*\rho_2)]<S$, which indicates that $(1-t^*)\rho_1+t^*\rho_2\notin\mathbb{E}_S[W]$.

\section{CLASSIFICATION OF QUANTUM IMAGINARITY WITNESSES}\label{sect3}

Set $F_S=\{\Tr[W\sigma] \mid W \in \mathbb{W}_S, \sigma \in \mathcal{R}\}$ and $D_S=\{\Tr[W\rho] \mid W \in \mathbb{W}_S, \rho \in \mathcal{D}\}$. From the above discussion, the set $F_S$ is classified into following three classes: (1) $F_S=\{s \in \mathbb{R} \mid 0 \leqslant s \leqslant S\}$ when $S$ is a positive real number; (2) $F_S=\{s \in \mathbb{R} \mid s \geqslant 0\}$ when $S$ denotes either $>$ or $\geqslant$; and (3) $F_S=\{0\}$ when $S=0$. Consequently, for every fixed $S$ we have a corresponding imaginarity-witness criterion: for each $\rho \in \mathcal{D}$, if there exists some $W\in\mathbb{W}_S$ such that $\Tr[W\rho]\in D_S\setminus F_S$, then $\rho$ is an imagianry state. To characterize the imaginarity witness $W$, we introduce $\mathbb{E}_S[W]$ as the set of all imagianry states that can be witnessed by $W$ in the setting $(\mathbb{W}_S, F_S, D_S)$, that is, $\mathbb{E}_S[W]=\{\rho \in \mathcal{D} \mid \Tr[W\rho] \in D_S \setminus F_S\}$. Hereafter, for the convenience of subsequent discussion, we denote the two cases in (2) as \(S\) = \(>\) and \(S\) = \(\geq\), respectively.

\begin{thm}\label{completeness}
The imaginarity-witness criterion given by the set $(\mathbb{W}_S, F_S, D_S)$ is complete for any $S$, namely,
\begin{equation}
\mathcal{D}\setminus\mathcal{R}=\bigcup_{W\in\mathbb{W}_S}\mathbb{E}_S[W].
\end{equation}
\end{thm}

\begin{proof}
Suppose $\rho$ is a imaginarity. Without loss of generality, we assume $\Im(\rho_{mn}) \neq 0$. We need to show that there are some $W\in\mathbb{W}_S$ such that $\Tr[W\rho] \in D_S \setminus F_S$.

(1) Set \(S=0\). For all \(1 \leqslant j < k \leqslant d\), define
\(W_{j,k}=\frac{\mathbf{i}}{2}\big(\ketbra{j}{k}-\ketbra{k}{j}\big).\)
By assumption, \(\Im(\rho_{mn}) \neq 0\), which yields
\(\Tr[W_{m,n}\rho]=\Im(\rho_{mn}) \neq 0.\)
Since \(W_{m,n}\in\mathbb{W}_0\), we therefore obtain \(\rho\in\mathbb{E}_0(W_{m,n})\). We have thus identified a Hermitian operator $W_{m,n}$ in $\mathbb{W}_0$, that detects the imaginarity of the quantum state $\rho$. So
\begin{equation*}
\mathcal{D}\setminus\mathcal{R}=\bigcup_{W\in\mathbb{W}_0}\mathbb{E}_0[W].
\end{equation*}

(2) Set $S$ = $\geqslant$. It follows from (1) that $\Tr[W_{m,n}] \neq 0$, so at least one of $\Tr[W_{m,n}\rho]$ and $\Tr[-W_{m,n}\rho]$ must be negative. Consequently, there exists some $W \in \mathcal{W}_\geqslant:=\{W_{j,k}, -W_{j,k} \mid 1 \leqslant j < k \leqslant d\}$ satisfying $\Tr[W\rho]<0$, and from $\mathcal{W}_\geqslant \subseteq \mathbb{W}_\geqslant$, we obtain $\mathcal{D}\setminus\mathcal{R}=\bigcup_{W\in\mathbb{W}_\geqslant}\mathbb{E}_\geqslant(W)$.

(3) Set $S$ = $>$. Let $W=\dfrac{s - \operatorname{Tr}\left[\widetilde{W}\rho\right]}{\Im(\rho_{mn})}W_{m,n}+\widetilde{W}$, where $s<0$ and $\widetilde{W}$ is a Hermitian operator such that $\Re\left(\widetilde{W}\right)$ is positive semi-definite and $\norm{\Re\left(\widetilde{W}\right)}_\infty>0$. Therefore, $W\in\mathbb{W}_>$ and $\Tr[W\rho]=s<0$. So, there exists a Hermitian operator in $\mathbb{W}_>$ that can detect the imaginarity of $\rho$.

(4) Let $S$ be a positive number. Similar to case (3), let $W=\dfrac{S+1 - \operatorname{Tr}\left[\widetilde{W}\rho\right]}{\Im(\rho_{mn})}W_{m,n}+\widetilde{W}$, where $\widetilde{W}$ satisfying that $\Re\left(\widetilde{W}\right) \geqslant 0$ and $\norm{\Re\left(\widetilde{W}\right)}_\infty\leqslant S$. So we have $W\in\W_S$ and $\Tr[W\rho] = S+1 > S$, i.e., $W$ detects the imaginarity of the state $\rho$.
\end{proof}

Having established completeness, we proceed the imaginarity-witness criteria for the four corresponding cases as follows.

(i) $(\mathbb{W}_S, F_S, D_S)$ with $S$ being a positive real number: a state $\rho\in\mathcal{D}$ is imaginary if and only if a $W\in\mathbb{W}_S$ such that either $\Tr[W\rho]<0$ or $\Tr[W\rho]>S$ exists.

(ii) $(\mathbb{W}_>, F_>, D_>)$: a state $\rho\in\mathcal{D}$ is imaginary if and only if there exists $W\in\mathbb{W}_>$ such that $\Tr[W\rho]<0$.

(iii) $(\mathbb{W}_\geqslant, F_\geqslant, D_\geqslant)$: a state $\rho\in\mathcal{D}$ is imaginary if and only if a $W\in\mathbb{W}_\geqslant$ such that $\Tr[W\rho]<0$ exists.

(iv) $(\mathbb{W}_0, F_0, D_0)$: a state $\rho\in\mathcal{D}$ is imaginary if and only if a $W\in\mathbb{W}_0$ such that $\Tr[W\rho]\neq0$ exists.

We call a imaginarity-witness criterion $(\mathbb{W}_S, F_S, D_S)$ finitely completable if all the imaginary states can be detected by a finite set of imaginarity witnesses in $\mathbb{W}_S$. That is, a finite set $\{W_k\}_{k=1}^n\subseteq\mathbb{W}_S$ such that
\begin{equation}
\mathcal{D}\setminus\mathcal{R}=\bigcup_{k=1}^n\mathbb{E}_S[W]
\end{equation}
exists. Otherwise, we call it finitely incompletable.

\begin{thm}
The imaginarity-witness set $(\mathbb{W}_S, F_S, D_S)$ is finitely completable if and only if $S=0$ or $S$ = $\geqslant$.
\end{thm}

\begin{proof}
The sufficiency can be observed from the proof of Theorem~\ref{completeness} that when  $S=0$ or $S$ = $\geqslant$, only finitely many witness operators are needed to detect all imaginary states; i.e., the imaginarity detection criterion is completable in these case.

For the necessity, we first show that for any positive real number $S$, $(\mathbb{W}_S, F_S, D_S)$ is finitely incompletable. For a finite subset $\mathcal{W}:=\{W_k\}_{k=1}^n$, let $\rho_\epsilon=(1-\epsilon)\pi_d+\epsilon\rho'$, where $\pi_d=\frac{\mathbb{I}_d}{d}$ is the maximally mixed state 
and $\rho'\in\mathcal{D}\setminus\mathcal{R}$ is some imaginary state. Denote $\eta_k\geqslant 0$ the smallest eigenvalue of the real part $\Re(W_k)$ of $W_k$ and 
$\xi=\max\limits_{1 \leqslant k\leqslant n}\{\Tr[W_k\pi_d]\}$, $M:=\max\limits_{1\leqslant k \leqslant n}\{\abs{\Tr[W_k\rho']}\}$. Let
\begin{equation*}
0 < \epsilon < \min\left\{\frac{\Tr[W_k\pi_d]-\eta_k}{M+\xi+1},\frac{S - \Tr[W_k\pi_d]}{M+\xi+1}\right\}.
\end{equation*}
Then
\begin{equation*}
\begin{aligned}
-M-\xi&\leqslant-\abs{\Tr[W_k\rho']}-\Tr[W_k\pi_d]\\
&\leqslant\Tr[W_k\rho'] - \Tr[W_k\pi_d]\\
&\leqslant\abs{\Tr[W_k\rho']}+\Tr[W_k\pi_d]\leqslant M+\xi.
\end{aligned}
\end{equation*}
Thus
\begin{equation*}
\begin{split}
\Tr[W_k\rho_\epsilon]&=(1-\epsilon)\Tr[W_k\pi_d]+\epsilon\Tr[W_k\rho'] \\
&=\epsilon\{\Tr[W_k\rho']-\Tr[W_k\pi_d]\} + \Tr[W_k\pi_d] \\
&>-(\Tr[W_k\pi_d]-\eta_k) + \Tr[W_k\pi_d] \\
&=\eta_k\geqslant 0
\end{split}
\end{equation*}
and
\begin{equation*}
\begin{split}
S - \Tr[W_k\rho_\epsilon]&=S-\Tr[W_k\pi_d]-\epsilon\{\Tr[W_k\rho']-\Tr[W_k\pi_d]\} \\
&>S-\Tr[W_k\pi_d]-(S-\Tr[W_k\pi_d]) \\
&=0,
\end{split}
\end{equation*}
which indicates that $\rho_\epsilon$ cannot be detected by any operator in the $\mathcal{W}$.

Now we show that $(\mathbb{W}_>, F_S, D_>)$ is also finitely incompletable. For any finite set $\mathcal{W}_> =\{W_k\}_{k=1}^n\subseteq \mathbb{W}_>$, we define $K=\max\limits_{1 \leqslant k \leqslant n}\{\abs{\Tr[W_k\rho']-\Tr[W_k\pi_d]}\}+1$ and $\zeta=\min\limits_{1 \leqslant k \leqslant n}\{\Tr[W_k\pi_d]\}$. If we set $0 < \epsilon < \min\left\{\dfrac{\zeta}{2K},1\right\}$, then $\Tr[W_k\rho_\epsilon]=\epsilon\{\Tr[W_k\rho']-\Tr[W_k\pi_d]\} + \Tr[W_k\pi_d]\geqslant\dfrac{\zeta}{2}>0$ for every $W_k$. That is, $\rho_\epsilon$ is a imaginary state whose imaginarity cannot be detected by the witnesses in $\mathcal{W}_>$.
\end{proof}

\begin{thm}
Let \(W_1\) and \(W_2\) be two imaginarity witnesses. The following assertions hold.

(1) Set $S$ = $>$ or $S$ = $\geqslant$. Then \(\mathbb{E}_S[W_1]=\mathbb{E}_S[W_2]\) if and only if there exists some \(r>0\) such that \(W_2 = r W_1\). More generally, \(\mathbb{E}_S[W_1]\subseteq \mathbb{E}_S[W_2]\) holds if and only if there exists a constant \(a>0\) and a positive semi-definite operator P such that \(W_1 = a W_2 + P\).

(2) Let $S$ be a positive real number. Then \(\mathbb{E}_S[W_1]=\mathbb{E}_S[W_2]\) if and only if either \(W_1=W_2\) or \(W_1+W_2 = S\mathbb{I}_d\).

(3) Set \(S=0\). Then \(\mathbb{E}_0[W_1]=\mathbb{E}_0[W_2]\) if and only if there exists a nonzero real number \(r\in\mathbb{R}\setminus\{0\}\) satisfying \(W_2 = r W_1\). Moreover, if no such nonzero real scalar r exists, then \(\mathbb{E}_0[W_1]\not\subseteq\mathbb{E}_0[W_2]\) and \(\mathbb{E}_0[W_2]\not\subseteq\mathbb{E}_0[W_1]\).
\end{thm}
\begin{proof}

(1) It follows from $\mathbb{E}_S[W]\neq\emptyset$ that $W$ has negative eigenvalues. For any Hermitian operator $W\in\mathbb{H}$, we define $\mathcal{D}_W=\{\rho\in\mathcal{D}\mid\Tr[W\rho]=0\}$ and $\mathcal{M}_W=\{X\in\mathrm{Mat}_d(\mathbb{C})\mid\Tr[WX]=0\}$. Lemma 1 in Ref.~\cite{LiPRA032422} proves that if $W$ admits negative eigenvalues, there exists a state $\sigma$ satisfying $\Tr[W\sigma]=0$, which implies $\mathcal{D}_W\neq\emptyset$. Given the condition $\mathbb{E}_S[W_1]=\mathbb{E}_S[W_2]$, we claim that $\mathcal{D}_{W_1}=\mathcal{D}_{W_2}$. Otherwise, without loss of generality, assume there exists a state $\rho\in\mathcal{D}_{W_1}$ and $\rho\notin\mathcal{D}_{W_2}$, i.e., $\Tr[W_1\rho]=0$ and $\Tr[W_2\rho]\neq 0$. If $\Tr[W_2\rho]<0$, then $\rho\in\mathbb{E}_S[W_2]$ yet $\rho\notin\mathbb{E}_S[W_1]$, contradicting our premise. If $\Tr[W_2\rho]>0$, choose an arbitrary state $\delta\in\mathbb{E}_S[W_1]$ which satisfies $\Tr[W_1\delta]<0$ and $\Tr[W_2\delta]<0$. Consider the convex combination $\rho_\epsilon=(1-\epsilon)\rho+\epsilon\delta$, $\epsilon\in(0,1)$. For sufficiently small $\epsilon$, 
$\Tr[W_1\rho_\epsilon]=\epsilon\Tr[W_1\delta]<0$ and $\Tr[W_2\rho_\epsilon]=(1-\epsilon)\Tr[W_2\rho]+\epsilon\Tr[W_2\delta]=\Tr[W_2\rho]-\epsilon\{\Tr[W_2\rho]-\Tr[W_2\delta]\}>0$. This yields $\rho_\epsilon\in\mathbb{E}_S[W_1]$ but $\rho_\epsilon\notin\mathbb{E}_S[W_2]$, another contradiction. We therefore conclude $\mathcal{D}_{W_1}=\mathcal{D}_{W_2}$. It thus follows from Lemma 2 in Ref.~\cite{LiPRA032422} that $\mathcal{M}_{W_1}=\mathcal{M}_{W_2}$. By definition, $\mathcal{M}_{W_1}$ and $\mathcal{M}_{W_2}$ are the orthogonal complement spaces associated with $W_1$ and $W_2$, respectively, each of dimension $d^2-1$. This implies $W_2=rW_1$, and one can further verify that $r$ is a positive real number. The other direction of the first statement is straightforward.

We now proceed to prove the second statement. Assume $\mathbb{E}_S[W_1]\subseteq\mathbb{E}_S[W_2]$. Set $t_1=\Tr[W_1]=\Tr[\Re(W_1)]$ and $t_2=\Tr[W_2]=\Tr[\Re(W_2)]$. Clearly, $t_1,t_2\geqslant 0$ as $W_1,W_2\in\mathbb{W}_S$. We prove the conclusion according to the following four cases.

(i) $t_1,t_2>0$. Notice that $\mathbb{E}_S[W_1]=\mathbb{E}_S\left[\frac{W_1}{t_1}\right]$ and $\mathbb{E}_S[W_2]=\mathbb{E}_S\left[\frac{W_2}{t_2}\right]$. Theorem 6 from Ref.~\cite{zhangPLA130135} guarantees that there exists a real number $0\leqslant\epsilon<1$ and a positive semi-definite operator $Q$ such that $\frac{W_1}{t_1}=(1-\epsilon)\frac{W_2}{t_2}+\epsilon Q$. Hence, $W_1=\frac{(1-\epsilon)t_2}{t_1}W_2+t_1\epsilon Q$. Setting $\frac{(1-\epsilon)t_2}{t_1}=a$ and $t_1\epsilon Q=P$ we complete the proof.

(ii) $t_1>0,t_2=0$. It is straightforward to verify that $\mathbb{E}_S[W_1]\subseteq\mathbb{E}_S[W_1+W_2]\subseteq\mathbb{E}_S[W_2]$. Analogously to case (i), we have $\frac{W_1}{t_1}=(1-\epsilon)\frac{W_1+W_2}{t_1}+\epsilon Q$, $\epsilon\in[0,1)$. Clearly, $\epsilon\neq 0$; otherwise, $W_2=\mathbf{0}$ is not a witness operator. Setting $\frac{1-\epsilon)}{\epsilon}=a$ and $t_1\epsilon Q=P$ we complete the proof.

(iii) $t_1=0,t_2>0$. From the chain of inclusions $\mathbb{E}_S[W_1]\subseteq\mathbb{E}_S[W_1+W_2]\subseteq\mathbb{E}_S[W_2]$, we obtain $\frac{W_1+W_2}{t_2}=(1-\epsilon)\frac{W_2}{t_2}+\epsilon Q$, $\epsilon\in[0,1)$. Similarly, if $\epsilon=0$, then $W_1=\mathbf{0}$ cannot be a witness. Hence, $\epsilon\neq 0$. Thus, $W_1+\epsilon W_2=t_2\epsilon Q$ and $\mathbb{E}_S[W_1]\subseteq\mathbb{E}_S[W_1+\epsilon W_2]\subseteq\mathbb{E}_S[W_2]$. On the other hand, $\mathbb{E}_S[W_1+\epsilon W_2]=\mathbb{E}_S[t_2\epsilon Q]=\emptyset$, which leads to a contradiction and this case cannot occur.

(iv) $t_1=t_2=0$. We demonstrate that the inclusion relation holds only if $W_2=aW_1$ for some scalar $a>0$. In fact, for small enough $\epsilon>0$, the chain of inclusions $\mathbb{E}_S[W_1+\epsilon\mathbb{I}_d]\subseteq\mathbb{E}_S[W_1]\subseteq\mathbb{E}_S[W_2]$ holds. Applying the reasoning established in case (ii), there exists a positive real number $a_{\epsilon}>0$ and a positive semi-definite $P_{\epsilon}$ satisfying $W_1+\epsilon\mathbb{I}_d=a_{\epsilon}W_2+P_{\epsilon}$. This identity enforces all diagonal entries of $P_{\epsilon}$ to be $\epsilon$. Additionally, since $P_{\epsilon}$  is positive semi-definite, the module of each off-diagonal entry is bounded above by $\epsilon$. Taking the limit $\epsilon\rightarrow0$, the operator $P_{\epsilon}$ converges to the zero operator, from which we conclude $W_2=aW_1$.

For the converse direction, suppose there exists a scalar $a>0$ and a positive semi-definite operator $P$ such that $W_1=aW_2+P$. Take any $\rho\in\mathbb{E}_S[W_1]$, we have $\Tr[W_1\rho]=a\Tr[W_2\rho]+\Tr[P\rho]<0$. Since $P$ is positive semi-definite, we always have $\Tr[P\rho]\geqslant 0$. This forces $\Tr[W_2\rho]<0$, which yields $\rho\in\mathbb{E}_S[W_2]$. Hence $\mathbb{E}_S[W_1]\subseteq\mathbb{E}_S[W_2]$.

(2) From Eq.~(\ref{leftright}) we have $\mathbb{E}_S[W_1] = \mathbb{E}_\geqslant[W_1] \bigcup\mathbb{E}_\geqslant[S\mathbb{I}_d-W_1]$ and $\mathbb{E}_S[W_2] = \mathbb{E}_\geqslant[W_2] \bigcup \mathbb{E}_\geqslant[S\mathbb{I}_d-W_2]$. Case 1: $\mathbb{E}_\geqslant[W_1]=\mathbb{E}_\geqslant[W_2]$ and $\mathbb{E}_\geqslant[S\mathbb{I}_d-W_1]=\mathbb{E}_\geqslant[S\mathbb{I}_d-W_2]$. In this case, it follows from the statement (1) that there exist scalars $r_1,r_2>0$ such that $W_2=r_1W_1$ and $S\mathbb{I}_d-W_2=r_2(S\mathbb{I}_d-W_1)$. This yields $r_1=r_2=r$, and consequently $W_2=rW_1$ with $r>0$. Case 2: $\mathbb{E}_\geqslant[W_1]=\mathbb{E}_\geqslant[S\mathbb{I}_d-W_2]$ and $\mathbb{E}_\geqslant[W_2]=\mathbb{E}_\geqslant[S\mathbb{I}_d-W_1]$. Then there exist scalars $r_1,r_2>0$ such that $S\mathbb{I}_d-W_2=r_1W_1$ and $S\mathbb{I}_d-W_1=r_2W_2$, i.e.
\begin{equation}\label{W1+W2}
\begin{cases}
r_1 W_1 + W_2 & = S\mathbb{I}_d, \\
W_1 + r_2 W_2 & = S\mathbb{I}_d.
\end{cases}
\end{equation}
If $r_1 r_2 \neq 1$, we obtain
\begin{equation*}
\begin{cases}
W_1 = \dfrac{r_2-1}{r_1r_2-1}S\mathbb{I}_d, \\
W_2 = \dfrac{r_1-1}{r_1r_2-1}S\mathbb{I}_d.
\end{cases}
\end{equation*}
In this case, $W_1$ and $W_2$ fail to be valid witness operators, which leads to a contradiction. Hence $r_1 r_2=1$. Substituting this relation into Eq.~(\ref{W1+W2}) yields $r_1=r_2=1$, and consequently $W_1 + W_2 = S\mathbb{I}_d$.

(3) On the one hand, if $\mathbb{E}_0[W_1]=\mathbb{E}_0[W_2]$, we conclude that $\mathcal{D}_{W_1}=\mathcal{D}_{W_2}$. Following the argument in (1), there exists a non-zero real number $r$ such that $W_2=rW_1$. On the other hand, suppose $W_2=rW_1$ with $r\neq 0$, we always have $\Tr[W_1\rho]=\neq 0$ if and only if $\Tr[W_2\rho]\neq 0$. Hence $\mathbb{E}_0[W_1]=\mathbb{E}_0[W_2]$.
\end{proof}

\section{Relation to Coherence Witnessing}\label{sect4}

We compare the imaginarity witnesses with coherence witnesses. A coherence witness $W$ distinguishes diagonal density matrices from non-diagonal ones.  
We indicate that the imaginarity witnesses have quite different properties from the coherence witnesses.  In Ref.~\cite{LiPRA032422} Li \textit{et al.} proved in Lemma 1 
that if the trace \(\operatorname{Tr}[W]\) of the observable \(W\) serves as prior knowledge for coherence detection, \(W\) necessarily admits negative eigenvalues. 
By contrast, Example 1 demonstrates that, for imaginarity detection, $W$ may be positive semi‑definite when the real‑part spectral norm $\|\Re(W)\|_{\infty}$ is taken as prior knowledge about $W$.

Li \textit{et al.} \cite{LiPRA032422} proved that the expectation value of an arbitrary Hermitian operator W over all incoherent states \(\delta\) obeys the following upper bound.,
\begin{equation}
0 \leqslant \Tr[W\delta] \leqslant \Tr[W].
\end{equation}
Furthermore, Zhu \textit{et al}. \cite{zhuJPA455208} derived both lower and upper bounds for the expectation value of a Hermitian operator $W$ over all incoherent states $\delta$,
\begin{equation}
\mu_{\min} \leqslant \Tr[W\delta] \leqslant \mu_{\max},
\end{equation}
where $\mu_{\min}$ and $\mu_{\max}$ are the minimum and maximum diagonal entries of the Hermitian operator $W$, respectively.

The $\norm{\Re(W)}_\infty$ and $\Tr[W]$ has the following quantitative relation,
\begin{equation*}
\Tr[W]=\Tr[\Re(W)]\geqslant\norm{\Re(W)}_\infty.
\end{equation*}
Let $\Re(W)=Q\Lambda Q^\top$ be the spectral decomposition of $\Re(W)$, where $\Lambda=\text{diag}\{\lambda_1,\lambda_2,\dots,\lambda_d\}$. Then $[\Re(W)]_{kk}=\sum\limits_{m=1}^d Q_{km}^2\lambda_m\leqslant\norm{\Re(W)}_\infty\sum\limits_{m=1}^d Q_{km}^2=\norm{\Re(W)}_\infty$, where $[\Re(W)]_{kk}$ denotes the $k$-th diagonal entry of $\Re(W)$. Hence $\mu_{\max}\leqslant\norm{\Re(W)}_\infty$. To sum up, we obtain $0 \leqslant \mu_{\max} \leqslant \norm{\Re(W)}_\infty \leqslant \Tr[W]$. We plot these bounds along the horizontal axis in the FIG.~\ref{fig3}.

\begin{figure}[htbp]
    \centering
    \includegraphics[width=0.4\textwidth]{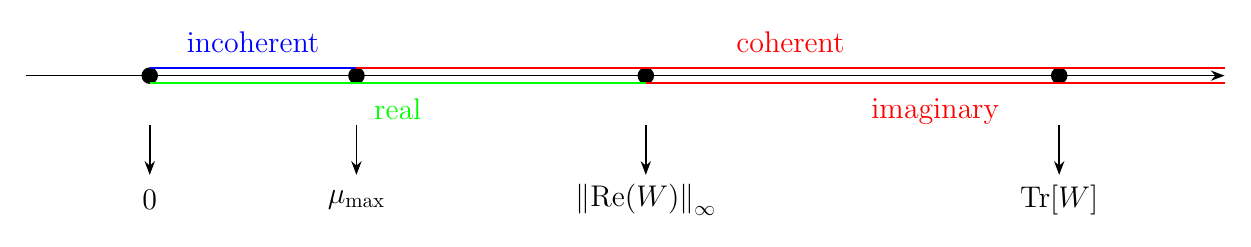}
    \caption{Relation between witnessing imaginarity and witnessing coherence using prior knowledge of observables.}
    \label{fig3}
\end{figure}

One observes that the expectation value of the Hermitian operator $W$ over any incoherent state falls within the interval $[0,\mu_{\max}]$, which corresponds to the blue line segment in the figure. If the expectation value $\Tr[W\rho]$ of $W$ evaluated on a quantum state $\rho$ lies in $(\mu_{\max},\infty)$, then $\rho$ is coherent, as indicated by the red region above the horizontal axis. Analogously, for quantum imaginarity, the expectation value of the Hermitian operator $W$ over all real quantum states is confined to $[0,\norm{\Re(W)}_\infty]$, marked by the green segment in the plot. When $\Tr[W\rho] \in (\norm{\Re(W)}_\infty,\infty)$, the quantum state $\rho$ has imaginarity, corresponding to the red region below the horizontal axis.

\section{CONCLUSION}\label{sect5}

We systematically investigate the fundamental properties of quantum imaginarity detection. We derive a rigorous analytical bound, which demonstrates that the maximal expectation value of an arbitrary imaginarity witness over all real quantum states is constrained by the spectral norm of the real part of the witness observable. This result unveils the intrinsic connection between the spectral characteristic of observables and the detection capability of witness operators, and provides a quantitative criterion to evaluate and optimize imaginarity detection schemes. Based on the derived spectral norm bound, we have further classified imaginarity witness operators into four exclusive categories and conduct a comprehensive comparative analysis of their core properties. We have clarified the completeness and finite completeness of each witness class, derived explicitly the universal conditions for different witness operators to achieve joint detection of common imaginary states,
and provided consistent identification of identical imaginary states. 

Recently, in  Ref. \cite{Liang04763} Liang \textit{et al}. redefined imaginarity witnesses and derived that the expectation value $\Tr[W\sigma]$ of a Hermitian operator $W$ over all real states $\sigma$ is bounded by the minimum eigenvalue $\lambda_{\min}(\Re(W))$  and the maximum eigenvalue $\lambda_{\max}(\Re(W))$ of the real part of $W$. Accordingly, a quantum state $\rho$ is an imaginary state if $\Tr[W\rho]\notin[\lambda_{\min}(\Re(W)),\lambda_{\max}(\Re(W))]$. This definition extends the scope of imaginarity witnesses, in contrast to the usual definition in which the real part $\Re(W)$ of $W$ is required to be positive semi‑definite. The imaginarity detection is improved by modifying the original framework, whereas we enhanced the detection capability by using the prior knowledge of the witness operators. It would also be appealing to extend our scheme to deal with other cases like entanglement witness. Moreover, our work lays a foundation for further exploring the operational value of quantum imaginarity as an independent quantum resource and promotes the practical exploitation of imaginarity-based advantages in quantum metrology, quantum cryptography and quantum information processing technologies.

\section*{Acknowledgments:}
S. M. Fei acknowledges the financial support from specific research fund of the
Innovation Platform for Academicians of Hainan Province.

\section*{Data availability}
No data were created or analyzed in this study.

\end{document}